\documentclass[12pt]{article}

\usepackage{hyperref}
\usepackage{amscd}
\usepackage{amsmath}
\usepackage{amsfonts}
\usepackage{amssymb}
\usepackage[english]{babel}

\newtheorem{theorem}{Theorem}[section]
\newtheorem{lemma}[theorem]{Lemma}
\newtheorem{proposition}[theorem]{Proposition}
\newtheorem{corollary}[theorem]{Corollary}
\newtheorem{definition}[theorem]{Definition}
\newtheorem{example}[theorem]{Example}
\newtheorem{remark}[theorem]{Remark}
\newtheorem{conjecture}{Conjecture}[section]

\newenvironment{proof}[1][Proof]{\textbf{#1.} }{\ \rule{0.5em}{0.5em}}

\newcommand{\bld}[1]{{\pmb #1}}
\newcommand{\ggt}[1]{{#1}{#1}^T}
\newcommand\bigzero{\makebox(0,0){\text{\huge0}}}

\newcommand{\Co}{\mathcal{C}}

\newcommand{\F}{\mathbb{F}}

\newcommand{\bc}{\bold{c}}
\newcommand{\bu}{\bold{u}}

\newcommand{\ba}{\bold{a}}
\newcommand{\bx}{\bold{x}}
\newcommand{\by}{\bold{y}}

\begin{document}

\title{Linear codes with arbitrary dimensional hull and pure LCD code.$\footnote{
 ymaouche@usthb.dz.}$ }

\author{Maouche Youcef}

\date{\small
${}^1$ Faculty of Mathematics, University of Science and Technology Houari Boumediene,\\ Bab Ezzouar, Algiers, Algeria.\\
}
\maketitle

\leftskip 0.8in
\rightskip 0.8in
\noindent
{\bf Abstract.} In this paper, we introduce a general construction of linear codes with small dimension hull from any non LCD codes. Furthermore, we show that for any linear code $\Co$ over
$\F_q$ ($q > 3$) with $dim(Hull(\Co))=h$ there exist an equivalent codes $\Co_j$ with $dim(Hull(\Co_j))=j$ for any integer $0\leq j \leq h$. We also introduce the notion of pure LCD code; an LCD code and all its equivalent are LCD; and construct an infinite family of pure LCD codes. In addition, we introduce a general construction of linear codes with one dimension hull.

\vskip 6pt
\noindent
{\bf Keywords.} LCD codes; Pure LCD codes; Hulls; Entanglement-assisted quantum error correcting codes.

\vskip 6pt
\noindent
2010 {\it Mathematics Subject Classification.} Primary 94B15, 94B05; Secondary 11T71.

\leftskip 0.0in
\rightskip 0.0in

\vskip 30pt

\section{Introduction}
The Euclidean hull of a linear code $\Co$ is the intersection of $\Co$ and its Euclidean dual $C^\perp$. A linear code $\Co$ is called $h$-dim hull if $dim(hull(C)) = h$. With this definition, a linear complementary dual (LCD) code is a $0$-dim hull code and an $[n,k]$ self-orthogonal code is $k$-dim hull. LCD codes have been widely applied in data storage, communications and cryptography \cite{boon,burn2006,carlet2018,tang17,carlet19,jin,yan,yang}. Massey \cite{massey92} gave the algebraic characterization of LCD codes, and showed that asymptotically good LCD codes exist. In \cite{tang17}, the authors show that any linear code over $\F_q\,(q > 3)$ is equivalent to an Euclidean LCD code and any linear code over $\F_{q^2}\,(q > 2)$ is equivalent to a Hermitian LCD code.

Linear code with small hull are also interesting due to its crucial role in checking permutation equivalence of two linear codes and determining the complexity of algorithms for computing the automorphism group of a linear code \cite{sendrier,li19,carlet19}. In \cite{li19}, the authors present some necessary and sufficient conditions that a linear codes and cyclic codes have one-dimensional hull and construct cyclic codes with one dimensional hull. In \cite{carlet19}, Claude et al. employ character sums in semi-primitive case to construct LCD codes and linear codes with one-dimensional hull from cyclotomic fields and multiplicative subgroups of finite fields. In \cite{carlet2018}, Hao prove that for a nonnegative integer $h$ satisfying $0 \leq h \leq n-1$, a linear $[2n, n]$ self-dual code is equivalent to a linear $h$-dimension hull code. Recently, there have been a lot of research works on linear codes with small hulls and its application in entanglementassisted quantum error-correcting codes (EAQECCs), the reader is referred to \cite{sok22,guenda2018,sendrier}.

The main goal of this manuscript is to extend the results in \cite{carlet19,hao} and construct an arbitrary dimensional hull from an existing one. We show that for any $[n,k,d]$ linear code over $\F_q$ $(q>3)$, with $h$ dimensional hull there exist a monomial equivalent codes $\Co_j$ to $\Co$ with $j$ dimensional hull for any $0\leq j\leq h$. More precisely, we introduce a general construction of small dimensional hull codes from any linear codes and arbitrary dimensional hulls from any self-orthogonal code. Furthermore, we introduce pure LCD code and show that such codes are very rare or do not exist over finite fields with characteristic $2$. In addition, we construct an infinite family of pure LCD code over finite fields with odd characteristics.

This paper is organized as follows. In Section 2, we give some preliminaries on linear codes and necessary and sufficient conditions for such code to be $h$ dimensional hull. In Section 3, we construct linear codes with small dimensional hull from any linear code and arbitrary dimensional hull from self-orthogonal codes. In Section 4, we introduce the notion of pure LCD codes and give a general construction of linear codes with one-dimensional hull.

\vskip 30pt
\section{Preliminaries}

Throughout this paper, let $\F_q$ be a finite field of order $q=p^m$ where $p$ is prime and $m$ is a positive integer. The multiplicative group of $\F_q$ is denoted by $\F_q^*$. An $[n, k, d]$ linear code $\Co$ over $\F_q$ is a linear subspace of $\F_q^n$ with dimension $k$ and minimum Hamming distance $d$. Let $\bx= (x_0, x_1,\cdots , x_{n-1})$, $\by = (y_0, y_1,\cdots, y_{n-1})\in \F_q^n$, their inner product is defined as usual
$$\langle \bx, \by\rangle= x_0y_0 + x_1y_1 +\cdots+ x_{n-1}y_{n-1} \in \F_q.$$
Two vectors $\bx$, $\by$ are called orthogonal if $\langle \bx, \by\rangle = 0$. For a linear code $\Co$ over $\F_q$, its dual code $\Co^\perp$ is the set of vectors orthogonal to every codeword of $\Co$ under the inner product i.e.,

$$\Co^\perp = \{\bx \in \F_q^n \;\,|\;\, \langle \bx, \by \rangle = 0, \;\,\; \forall \by \in \Co\}.$$
A code $\Co$ is called self-orthogonal if $\Co \subseteq \Co^\perp$, and it is called self-dual if $\Co^\perp = \Co$. The hull of a linear code $\Co$ is defined by
$$\text{hull}(\Co) = \Co \cap \Co^\perp.$$
A linear code $\Co$ is called $h$-dim hull if $dim(hull(C)) = h$. With this definition, a linear complementary dual (LCD) code is a $0$-dim hull code and $[n,k]$ self-orthogonal code is $k$-dim hull. For any $\ba = (a_1,\cdots, a_n) \in \F^n_q$ and permutation $\sigma$ of $\{1, 2,\cdots, n\}$, we define $\Co_\ba$ and $\sigma (\Co)$ as the following linear codes
$$\Co_\ba = \{ (a_1c_1,\cdots, a_nc_n) \;:\; (c_1,\cdots, c_n) \in \Co \},$$
and
$$\sigma (\Co) = \{(c_{\sigma (1)},\cdots , c_{\sigma (n)}) \;:\; (c_1,\cdots , c_n) \in \Co \}.$$
Two codes $\Co$ and $\Co^\prime$ over $\F_q$  are called monomial equivalent if $\Co^\prime =\sigma (\Co_\ba)$ for some permutation $\sigma$ of $\{1, 2, \cdots, n\}$ and $\ba \in (\F^*_q)^n$. Let $\Co_1$ and $\Co_2$ be two codes over the same field $\F_q$, and let $G_1$ be a generator matrix for $\Co_1$. Then $\Co_1$ and $\Co_2$ are monomial equivalent if and only if there is a monomial matrix $M$ (a square matrix with exactly one nonzero entry in each row and column) so that $G_1M$ is a generator matrix of $\Co_2$.

Let $\Co$ be an $[n,k]$ linear code. If $\Co$ is the direct sum of subspaces $U_i$ for $1\leq i \leq k $ which are mutually orthogonal, then we shall say that 
$\Co$ is the orthogonal sum of the $U_i$ and use the symbol 
$$\Co =U_1 \perp \cdots \perp U_k. $$
\begin{theorem}\label{theorem1}
Let  $\Co$ be an $[n,k]$ linear code over $\F_q$ where $q$ is odd and $Hull(\Co)=h$. Then there exist code words $\bc_1,\cdots,\bc_k \in \Co$ such that
$$ \Co=<\bc_1>\perp <\bc_2>\perp \cdots \perp <\bc_k>$$
and $\Co$ is LCD if and only if $\bc_i^2 \not=0$ for $1\leq i \leq k$. Furthermore,
$$Hull(\Co)=<\bc_{j_1}>\perp <\bc_{j_2}>\perp \cdots \perp <\bc_{j_l}> \quad  \bc_{j_r}^2=0 \;\: \forall 1\leq r\leq l.$$
\end{theorem}

\begin{proof}
Following \cite{artin}, we regard $\Co$ with the inner product as a subspace of finite geometry. Therefore, the proof follows from [\cite{artin}, Theorem 3.7].
\end{proof}

\begin{corollary}\label{co1}
Let $\Co$ be an $[n, k]$ linear code over $\F_q$ with $dim(Hull(\Co)) = h$ and $q$ is odd. There exists a generator matrix $G$ of $\Co$ such that $GG^T$ is a diagonal matrix.
\end{corollary}

\begin{proof}
The proof is immediate from Theorem \ref{theorem1} by choosing an orthogonal basis for $\Co$.
\end{proof}

\begin{theorem}[\cite{Massey}]\label{th23}
If $G$ is a generator matrix for the $[n, k]$ linear code $\Co$, then $\Co$ is an LCD code if and only if, the $k \times k$ matrix $GG^T$ is nonsingular.
\end{theorem}

The following proposition was proven in \cite{guenda2018}. Theorem \ref{theorem1} and Corollary \ref{co1} give another proof for odd $q$, however, we give a third proof with a direct approach.
\begin{proposition}[\cite{guenda2018}]\label{guenda}
Let $\Co$ be an $[n, k, d]_q$ linear code with parity check matrix $H$ and generator matrix G. Then, $rank(HH^T)$ and $rank(GG)$ are independent of $H$ and $G$ so that
$$rank(HH^{T})= n-k-dim(hull(C^\perp)) = n-k-dim(hull(C)),$$
and
$$rank(GG^T)= k-dim(hull(C^\perp)) = k-dim(hull(C)).$$
\end{proposition}

\begin{proof} Since $Hull(C) = Hull(C^\perp)$, the second equalities are obvious. We have 
\begin{align*}
Hull(C)&=\left\{ x \in \F_q^n : x \in C, x \in C^\perp \right\} \\
&=\left\{ x \in \F_q^n :  x \in C, xG^T=0 \right\} \\
&=\left\{ yG  : y \in \F_q^k  \;\:\vert\;\:  yGG^T=0 \right\}. 
\end{align*}
Note that $\#Hull(C)=\#Ker(GG^T)$, hence $dim(Hull(C))=null(GG^T)$. Therefore
\begin{align*}
k&=rank(GG^T)+nul(GG^T)\\
&=rank(GG^T)+dim(Hull(C)).
\end{align*}
This completes the proof for the second statement. Since $G$ is a parity check matrix of $C^\perp$, the first statement can be proven by a similar argument.
\end{proof}

\vspace*{2mm}
The next theorem introduces a general construction of LCD code from any linear code over $\F_q$ ($q>3$).

\begin{theorem}[\cite{carlet2018}]\label{carlthe}
Let $q$ be a power of a prime with $q > 3$ and $\Co$ be an $[n, k, d]$ linear code over $\F_q$. Then, there exists $\bold{a}=(a_1,\cdots, a_n) \in \F^n_q$ with $a_j \not= 0$ for any $1 \leq j \leq n$ such that $\Co_\ba$ is an LCD code.
\end{theorem}

Let $M$ be an $k\times k$ matrix over $\F_q$ and $\bu$ be a vector in $\F_q^k$. Following \cite{carlet19}, we denote by $diag_k[\bu]$ the diagonal $k\times k$ matrix whose elements on the diagonal are $u_1,\cdots, u_k$. Let $I=\{i_1,\cdots,i_l\}$ be a subset of $\{1,\cdots,k\}$, we define the submatrix $M_I$ of $M$ obtained by deleting the $i_1,\cdots,i_j$-th rows and columns of $M$. Denote $M_I = 1$ if $I = \{1, 2,\cdots,k\}$ and $M_\emptyset = M$. The following lemmas will be used in the sequel.

\begin{lemma}[\cite{carlet2018}]\label{matrixlemma} Let $M$ be a $k\times k$ matrix over $\F_q$ and $t$ an integer with $0 \leq t \leq k-1$. Suppose that $det(M_I) = 0$ holds for any subset $I$ of $\{1, 2,\cdots,k\}$ with $0\leq \# I \leq t$. Then, for any $1\leq j \leq t + 1$ and every word $\bu \in \F_q^k$ of Hamming weight $j$, denoting its support by $J$, we have:
$$ det(M+diag_k(\bu))=\left(\prod_{i \in J} u_i\right) det(M_I) .$$
\end{lemma}
\begin{lemma}
Let $M$ be a non-singular $k\times k$ matrix over $\F_q$. For any vector $\bu \in \F_q^k/\{0\}$ with Hamming weight $j$ and support $J$, we suppose that $det(M_I) = 0$ holds for any subset $I$ of $J$ with $1\leq \# I < j$. Then
$$ det(M+diag_k(\bu))=det(M)+ \left(\prod\limits_{i \in J} u_i\right) det(M_J).$$
\end{lemma}

\begin{proof} We prove this statement by induction on $j$. For any vector $\bu$ of Hamming weight $1$, denoting by $i_1$ the position of its only nonzero coordinate, we have
$$det(M + diag_k(\bu)) = det(M) + u_{i_1} det(M_{\{i_1\}}). $$
Thus, the statement is true for $j =1$. Assume the statement is true for any vector $\bu$ with Hamming weight  $j=1,2,\cdots,s \leq k-1$. Let $\bu$ be any vector with Hamming weight $s+1$
and $support(\bu)=\{i_1,i_2,\cdots,i_{s+1} \}=J$. Let $\bu^\prime$ be the vectors with $support(\bu^{\prime})=\{i_1,i_2,\cdots,i_{s}\}=J^\prime$. Then
$$ det(M+diag(\bu))=det(M+diag_k(\bu^\prime))+ u_{i_{s+1}}det\left(M_{\{i_{s+1}\}} + diag_{k-1}(\bu^{\prime})\right).$$
Let $M^\prime=M_{\{i_{s+1}\}}$. Observe that for any subset $I$ of $J^\prime$ with $0\leq \#I \leq \#J^\prime-1 $ we have $det(M^\prime_I)=det(M_{I\cup{i_1}})=0$. Applying Lemma \ref{matrixlemma}, we get
\begin{align*}
det(M+diag(\bu))&=det(M+diag_k(\bu^\prime))+ \left(\prod_{i \in J} u_{i}\right) det\left(M_{J}\right).
\end{align*}
By induction assumption we have 
$$ det(M+diag(\bu^\prime))=det(M)+  \left(\prod_{i \in J^\prime} u_{i}\right) det\left(M_{J^\prime}\right)=det(M).$$
This implies
$$det(M+diag_k(\bld{u}))=det(M)+ \left(\prod\limits_{i \in J} u_i\right) det(M_J) .$$
This completes the proof.
\end{proof}

\vskip 30pt
\section{Arbitrary Hull dimension }

In this Section, we construct linear codes with small dimensional hull from any linear code and arbitrary dimensional hull from self-orthogonal codes. The following lemmas are used in the sequel.

\begin{lemma}\label{lemma1}
Let $q$ be a prime power with $q > 3$ and $\Co$ be an $[n,k]$ linear code over $\F_q$ with  $dim(hull(\Co))=h$ and a generator matrix $G = [I_k : P]$. Let $\bold{a}=(a_1,\cdots,a_n) \in \F_q^n$ with $a_j=1$ for any $1 \leq j \leq n$ except for one entry $a_i\not=0,1$ with $0\leq i\leq k$. Then  
$$h-1 \leq dim( Hull(\Co_\ba)) \leq h+1.$$
\end{lemma}

\begin{proof}
Let $G^\prime$ be a generator matrices of $\Co_\ba$ by multiplying $j$-th column of $G$ by $a_j$ for $j \in \{1, 2,\cdots, n\}$. Let $M=GG^T$ and $\bu\in \F_q^n$ be a vector of length $n$ such that $u_j = a^2_j-1$ for $1 \leq j \leq n$. Then $\bu $ is a one weight and 
$$ rank(G^\prime {G^\prime}^T)=rank(M+ diag(\bu)).$$ 
Since, $\bu$ is a one weight vector then adding $diag(\bu)$   effect only one entry of $M$. Then the span of $M$ will increase or decrees by a vector. Therefore
$$ rank(M)-1\leq  rank(\ggt{G^\prime}) \leq rank(M)+1.$$
Applying Proposition \ref{guenda}, we get 
$$ k-dim(Hull(\Co))-1\leq  k-dim(Hull(\Co_{a})) \leq k-dim(Hull(\Co))+1.$$
This implies
$$h-1 \leq dim( Hull(\Co_a)) \leq h+1.$$
This complete the proof.
\end{proof}

\begin{lemma}\label{lemma2}
Let $q$ be a prime power with $q > 3$ and $\Co$ be an $[n,k]$ linear code over $\F_q$ with  $h$-hull. Then there exist $\ba=(a_1,\cdots,a_n) \in \F_q^n$ with $a_j \not=0$ for any $1 \leq j \leq n$ such that $dim(Hull(\Co_\ba))=h-1$.
\end{lemma}

\begin{proof}
Without loss of generality we may assume that the generator matrix of $\Co$ is in standard form.  From the proof of Theorem \ref{carlthe} \cite{carlet19}, there exist $\ba^\prime=(a_1,\cdots,a_k,1,\cdots,1) \in \F_q^n$ with $a_j \not =0$ for $0 \leq j \leq k$ and $dim(Hull(\Co_{\ba^\prime}))=0$. Let $\ba_1=(a_1,1,\cdots,1)$ and $\Co_1=\Co_{\ba_1}$, then by Lemma \ref{lemma1} 
 $$ dim(Hull(\Co_1))=h-1 \quad \text{ or } \quad dim(Hull(\Co_1)) \geq h.$$
If $dim(Hull(\Co_1))=h-1$, then the proof is completed. If $dim(Hull(\Co_1)) \geq h$, let $\ba_2=(a_1,a_2,1,\cdots,1)$ and $\Co_2=\Co_{1_{(1,a_2,1\cdots)}}=\Co_{\ba_2}$. Applying Lemma \ref{lemma1} again
 $$ dim(Hull(\Co_{2}))=h-1 \quad \text{ or } \quad dim(Hull(C_{2})) \geq h$$
and using the same argument again. If $dim(Hull(C_{\ba_i}))\ge h$ for all $1 \leq i\leq k$, where $\ba_1=(a_1,1,\cdots,1)$, $\ba_2=(a_1,a_2,1,\cdots,1),\: \cdots,\:\ba_k=(a_0,a_1,\cdots,a_k,1,\cdots,1)=a^\prime$ this will contradict the existence of $a^\prime$ with $dim(Hull(\Co_{a^\prime}))=0$. Therefore, there exist $\ba=(a_1,\cdots,a_n) \in \F_q^n$ with $a_j\not=0$ for $ 1 \leq j\leq n$ and $dim(Hull(C_a))=h-1$.
\end{proof}

\begin{theorem}\label{theorm31}
Let $q$ be a power of a prime with $q > 3$ and $\Co$ be $[n, k, d]$-linear code over $\F_q$ with $dim(Hull(\Co))=h$. Then, there exists $\Co_j$ codes with the same parameters as $\Co$ and
$$ dim(Hull(\Co_j))=h $$
for any $0\leq j \leq h$.
\end{theorem} 
\begin{proof} We prove this theorem by induction on $j$. Obviously, the statement holds for $j=h$. Assume, by induction, that the statement holds for all $h, h-1, \cdots,j$. Let $\Co_j$ be an $[n,k,d]$ linear code with $dim(Hull(\Co_j))=j$. By Lemma \ref{lemma2} there exist a monomial equivalent code $\Co_{j-1}$ with $dim(Hull(C_{j-1}))=dim(Hull(C_{j}))-1=j-1$. This completes the proof.
\end{proof}

\begin{corollary}\label{corollaryself}
Let $q$ be a power of a prime with $q > 3$ and $\Co$ be an $[n, k, d]$ self-orthogonal (self-dual) code over $\F_q$. Then, there exist linear codes with the same parameters as $\Co$ with arbitrary Hull.
\end{corollary}

\begin{theorem}[\cite{ba2019,st2006}]\label{theoremlong}
There exist long $q$-ary self-dual codes which meet the Gilbert-Varshamov bound for odd $q$.
\end{theorem}

Combine  Theorem \ref{theoremlong} and Corollary \ref{corollaryself} we get the following results.

\begin{theorem}\cite{hao}
Let $h$ be a fixed positive integer. There exist long $h$-hull codes which meet the Gilbert-Varshamov bound.
\end{theorem}

Now we give a construction of maximal-entanglement EAQECCs from any linear code over $\F_q$ with $q>3$. Our result significantly improve the results in \cite{carlet19,hao}.

\begin{proposition}[\cite{burn2006}]
Let $\Co$ be a classical $[n, k, d]$ linear code over $\F_q$. Then there exist $[[n, k-dim(Hull(C)), d; n-k-dim(Hull(C))]]$ EAQECCs over $\F_q$. Further, if $\Co$ is MDS then the EAQECCs is also MDS.
\end{proposition}

\begin{corollary}
Let $\Co$ be an $[n, k, d]$ linear code over $\F_q$ with $q > 3$ and $dim(Hull(\Co))=h$. Then there exist $[[n, k-l, d; n-k-l]]_q$ EAQECCs for any $1\leq l\leq h$.
\end{corollary}

\begin{example} Let $\Co$ be $[7, 3, 2]$ self-orthogonal linear code over $\F_5$ with the following generator matrix
$$ G= \begin{pmatrix}
1 &0 &0 &0 &0 &2 &0 \\
0 &1 &0 &2 &2 &0 &4 \\
0 &0 &1 &1 &3 &0 &3
\end{pmatrix} $$ 
Let $a_0=(2,2,2,1,1,1,1)$, $a_1=(2,2,1,1,1,1,1)$ $a_2=(2,1,1,1,1,1,1)$ $a_3=(1,1,1,1,1,1,1)$. Then $\Co_{a_i}$ are linear codes with the same parameters as $\Co$ and 
$$dim(Hull(\Co_{a_i})) =i, $$
for $0\leq i\leq 3$.

\end{example}

\section{One Dimensional Hull and Pure LCD code}

In Section 3 we showed that for any $h$-hull linear code and any integer $j$ with $0<j\leq h$, there exist an equivalent code $\Co_j$ with $dim(Hull(\Co_j))=j$. A natural question rise; is this operation invertible? For any LCD code $\Co$ with parameters $[n,k,d]$, does a monomial equivalent code $\Co_1$ to $\Co$ with $dim(Hull(\Co))=1$ exist?. It turns out that, there exist linear codes where all their monomial equivalent codes are LCD. We introduce the following definition.
\begin{definition}
A linear code $\Co$ is called \textbf{ Pure LCD} code if and only if all its monomial equivalent codes are LCD. i.e
$$ dim(Hull(\Co_a))=0, \quad \text{ for all }\,  \ba\in {(\F_q^*)^n}.$$
\end{definition}

\begin{remark}
Note that the hull of a linear code over a finite field $\F_q$ with $q=2$ or $q=3$ is an invariant of equivalent codes. Therefore, a linear code $\Co$ over $\F_q$ ($q \in \{2,3\}$) is pure LCD if and only if $\Co$ is LCD. 
\end{remark}

The following example shows that pure LCD code over $\F_q$ ($q>3$) does exist.

\begin{example} Let $\Co$ be the code over $\F_5$ with generator matrix
$$ G=\begin{pmatrix}
1 &0 &0 &0 &0 &4 \\
0 &1 &0 &2 &4 &0 \\
0 &0 &1 &0 &0 &3
\end{pmatrix}.$$
Using Magma we can check that all monomial equivalent codes to $\Co$ are LCD codes. Hence, $\Co$ is pure LCD code.
\end{example}

\begin{theorem}\label{th41}
Let $q$ be a prime power with $q>3$ and $\Co$ be $[n, k, d]$ LCD code over $\F_q$ with generator matrix $G = [I_k :
P]$. Let $M=GG^T$. Assume that there exist $1\leq i\leq k$ with 
$-det(M)/det(M_{\{i\}})+1$ is a nonzero square, then $\Co$ is not pure LCD code. Furthermore, let $\ba=(a_1,\cdots,a_n)$ with $a_i^2=-det(M)/det(M_i)+1$ and $a_j=1$ for all $ 1\leq j\leq n$ and $j\not=i$, then $\Co_\ba$ is one dimensional hull. 
\end{theorem}

\begin{proof}
Let $G^\prime$ be a generator matrix of $\Co_{a}$ obtained from
$G$ by multiplying its $i$-th column by $a_i$. Then  
$$det(G^\prime {G^\prime}^T)=det(M)+(a_i^2-1)det(M_{\{i\}})=0.$$
Therefore, $\Co_{a}$ is not LCD by Theorem \ref{th23}. Furthermore  
$$rank(G^\prime {G^\prime}^T)=rank(M+diag(0,\cdots,0, a_i^2-1,0,\cdots,0)).$$
By Lemma \ref{lemma1}, $rank(\ggt{G^\prime})=k$ or $=k-1$. Since $det(\ggt{G^\prime})=0$ then $rank(\ggt{G^\prime})=k-1$ and $\Co_{a}$ is one dimensional hull by Proposition \ref{guenda}.
\end{proof}

\begin{remark} Using a computer we check that most codes that are not pure LCD codes meet the requirement of Theorem \ref{th41}. Therefore, combining Theorem \ref{theorm31} and Theorem \ref{th41} is very efficient for constructing one-dimensional Hull codes from any linear codes.
\end{remark}

\begin{lemma}\label{lemma46}
Let $q = 2^t$ with $t>1$ and $\Co$ be $[n, k]$ linear code over $\F_q$ with generator matrix $G = [I_k :
P]$. Let $M=GG^T$. Assume that there exist $\bu \in \F_q^n$ with $u_i \in \F_q/\{1\}$ and 
$$det(M+diag(\bu))=0.$$
Then $\Co$ is not pure LCD code.   
\end{lemma}
\begin{proof} Let $\ba=(a_1,\cdots,a_n)$ where $a_i^2-1=u_i$ for all $1\leq i\leq n$. Then $M+diag(\bu)$ is a generator matrix of $\Co_\ba$. The rest of the proof follows from Proposition \ref{guenda} and the fact that all non-zero element of $\F_q$ are square.
\end{proof}
\vspace{1.5mm}

In \cite{hao}, the author present a week condition on the existence of pure LCD code over $\F_{2^t}$. In the next theorem we give another conditions and one conjecture.

\begin{theorem}\label{theorem42} Let $q = 2^t$ with $t>1$ and let $\Co$ be $[n, k, d]$ linear code over $\F_q$. Assume one of the following conditions holds
\begin{itemize}
\item[•] $det(M_j)\not=0$ and $det(M_j)\not=det(M)$ for some integer $1\leq j\leq k$.
\item[•] There exist a subsets $J\subseteq \{1,\cdots,n\}$ with $det(M_J)\not=0$ and $det(M_I)=0$ for any subset $I$ of $J$ with $1\leq I < \#J$.
\end{itemize}
Then, $\Co$ is not pure LCD code and there exists $\ba = (a_1, \cdots, a_n) \in \F^n_q$ with $a_j \not= 0$ for any $1 \leq j \leq n$ such that $\Co_a$ is one dimensional hull.
\end{theorem}

\begin{proof} If $\Co$ is not an LCD code, the result holds from Theorem $\ref{theorm31}$. We assume that $\Co$ is LCD code and let $G=[I_k:P]$ be its generator matrix. Let $M=\ggt{G}$, then $det(M) \not=0$. 
\begin{itemize}
\item[] {\bf Case 1:} If $det(M_j)\not=0$ and $det(M_j)\not=det(M)$ for some integer $1\leq j\leq k$. The results holds from Theorem \ref{th41} and the fact that all element of $\F_q$ are squares.
\item[] {\bf Case 2:} Note that, there exist a subsets $J\subseteq \{1,\cdots,n\}$ with $det(M_J)\not=0$ and $det(M_I)=0$ for any subset $I$ of $J$ with $1\leq I < \#J$. Let $\bu=(u_1,\cdots,u_n)$ with $u_i=0$ for $i \not \in J$. Using Lemma \ref{matrixlemma} we get
\begin{align*}
det(M+diag(\bu))&=det(M)+\left(\prod_{i \in J}u_i\right) det(M_J)\\
&=det(M)+u_{j_1}u_{j_2}\left(\prod_{i \in J \setminus \{j_1,j_2\}}u_i\right) det(M_J).
\end{align*}
By choosing $u_j \in \F_q\setminus \{0,1\}$ for $j \in J \setminus \{j_1,j_2\}$ and 
$$
u_{j_1}=\begin{cases}
u_{j_2}^{-1} & \text{ if } \left(\prod_{i \in J/\{j_1,j_2\}}u_i\right) \frac{det(M_J)}{det(M)}=1, \\
\left(\prod_{i \in J/\{j_1,j_2\}}u_i\right)^2 \left(\frac{det(M_J)}{det(M)}\right)^2 & \text{ otherwise}.
\end{cases}
$$
Because $u_i\not =1$ for all $1 \leq i\leq n$ and $det(M+diag(\bu))=0$, the result holds from Lemma \ref{lemma46} and Theorem \ref{theorm31}.
\end{itemize}
\end{proof}

\begin{remark}
Theorem \ref{theorem42}, give very powerful restriction on the existence of pure LCD code over $\F_{2^t}$. In fact, linear $[n,k]$ codes with generator matrix $GG^T=I_k$ are the only codes we found that escape this two condition and its not difficult to prove such codes are not pure LCD code. Therefore we introduce the following conjecture.
\end{remark}

\begin{conjecture}
Let $q = 2^t$ with $t>1$ and $\Co$ be an $[n,k]$ linear code over $\F_q$. Then $\Co$ is not pure LCD code.
\end{conjecture}

In the following theorem we shows that there exist infinitely many pure LCD codes.

\begin{theorem} Let $q$ be an odd prime power. Let $\Co$ be an $[n,k]_q$ linear code with the generator matrix $G = [I_k:I_k]$. If $-1$ is not a square in $\F_q$ then $\Co$ is pure LCD code.
\end{theorem}
\begin{proof} Let $\ba=(a_1,\cdots,a_n) \in \F_q$ with $a_j\not=0$ for any $1\leq j \leq n$. Let $G^\prime$ be the generator matrix of $\Co_\ba$ obtained from $G$ by multiplying its $j$-th column by  $a_j$ for $j \in \left\{ 1, 2,\cdots, n \right\}$. Then 

$$\ggt{G^\prime}=\begin{pmatrix}
  \begin{matrix}
  a_1^2+a_{k+1}^2 &  \\
   & a_2^2+a_{k+2}^2
  \end{matrix}
  &  & \bigzero \\
  \bigzero &  &
  \begin{matrix}
   \ddots &  \\
   & a_k^2+a_{2k}^2
  \end{matrix}
\end{pmatrix}.
$$
Since, $-1$ is not a square in $\F_q$ then $a_i^2+a_{i+k}^2=0$ if and only if $a_i=0$ and $a_{i+k}=0$. Hence $rank(\ggt{G^\prime})=k$. Therefore, by Proposition \ref{guenda} 
$$ dim(Hull(\Co_a))=0.$$
This completes the proof.
\end{proof}

\section*{Concluding remarks}
In this paper, we presented a general construction of linear code with a small dimensional hull from any linear codes. We showed that for any linear code with $h$-hull dimension, there exists an equivalent code with $h$-hull for any $0\leq j\leq c$. In particular, for any $[n,k,d]$ self-orthogonal code there exists a code with the same parameters and with an arbitrary dimensional hull. We also introduce the notion of pure LCD code. Furthermore, we give sufficient conditions for the existence of one dimensional  $[n,k,d]$ code from another code with the same parameters. 

Finally, we present a family of pure LCD codes over finite fields with odd characteristics and very week conditions for the existence of pure LCD code over finite fields with even characteristics and rise a conjecture for this case. An interesting extension to this work would be to give a necessary and sufficient condition for the existence of pure LCD code over finite fields with odd characteristics.

\end{document}